\documentclass[%
 reprint,
 amsmath,amssymb,
 aps,
]{revtex4-1}

\usepackage{graphicx}
\usepackage{dcolumn}
\usepackage{bm}

\usepackage{color}
\usepackage{amsthm}

\newtheorem{theorem}{Theorem}

\newtheorem{definition}{Definition}

\newcommand{\be}{\begin{eqnarray}}
\newcommand{\ee}{\end{eqnarray}}
\newcommand{\nn}{\nonumber}
\newcommand{\bn}{\begin{enumerate}}
	\newcommand{\en}{\end{enumerate}}
\newcommand{\bl}{\begin{align}}
\newcommand{\el}{\end{align}}

\def\a{\alpha}

\def\d{\delta}
\def\eps{\epsilon}

\def\th{\theta}

\def\la{\lambda}

\def\r{\rho}

\def\p{\psi}

\def\La{\Lambda}
\def\P{\Psi}

\def\iff{\Longleftrightarrow}

\def\<{\langle}
\def\>{\rangle}

\def\jmath{{j}}

\begin{document}


\title{Generalized Coherence Concurrence and Path Distinguishability}

\author{Seungbeom Chin}
\email{sbthesy@skku.edu}
\affiliation{College of Information and Communication Engineering, Sungkyunkwan University, Suwon 16419, Korea}



\begin{abstract}
We propose a new family of coherence monotones, named the \emph{generalized coherence concurrence} (or coherence $k$-concurrence), which is an analogous concept to the generalized entanglement concurrence. The coherence $k$-concurrence of a state is nonzero if and only if the coherence number (a recently introduced discrete coherence monotone) of the state is not smaller than $k$, and a state can be converted to a state with nonzero entanglement $k$-concurrence via incoherent operations if and only if the state has nonzero coherence $k$-concurrence. We apply the coherence concurrence family to the problem of wave-particle duality in multi-path interference phenomena. We obtain a sharper equation for path distinguishability (which witness the duality) than the known value and show that the amount of each concurrence for the quanton state determines the number of slits which are identified unambiguously.  
\begin{description}
	\item[PACS numbers] 03.65.Ta, 03.67.Bg, 03.67.Mn
\end{description}
\end{abstract}


\maketitle


\section{INTRODUCTION}

The superposition principle is a key aspect of quantum physics, and representative quantum properties such as entanglement are understood best under the framework of superposition. We can say that a state is nonclassical (quantum) if and only if it is a superposition of some classical state \cite{killoran}. In the resource theory of  coherence, classical states are in some fixed orthonormal basis set $\{ |i\>\}_{i=1}^{d}$, and the coherence is a basis-dependent quantity.

 Since the quantitative formulation for measuring the amount of coherence is presented in \cite{baum}, the coherence resource theory flourished in diverse aspects, e.g., finding new measures and monotones of coherence \cite{yuan, winter, napoli, piani, tan, chitambar}, understanding the relation of coherence with other correlations\cite{strel, adesso, roga, ma, mmarv, marv, rana, rana2}, dynamics of coherence\cite{bromley, mani, puchala, singh, mondal}, and quantum thermodynamical approaches\cite{lost, cwik, naras} (see \cite{SP} for an up-to-date review on the coherence resource theory).
 
 One of the interesting topics in this area is to understand the connection between coherence and entanglement theory. In \cite{strel} it was shown that a nonzero coherent state can be used to create  entanglement. Killoran $et$ $al.$ \cite{killoran} generalized this process, and provided a framework for converting nonclassicality (including coherence) into entanglement. While discussing the conversion theorem, the authors presented an analogous concept to the Schmidt rank of entangled pure states, which is the $coherence$ $rank$ of pure states. This concept was generalized to mixed state case in \cite{Chin}, where the $coherence$ $number$ $r_C(\r)$ of a mixed state $\r$ was introduced along the Schmidt number of entangled mixed states. It is proved there that a state can be converted to an entangled state of nonzero $k$-concurrence (the $k$-th member of the generalized concurrence monotone family \cite{gour}) if and only if the coherence number of the state is not smaller than $k$. 

In this paper we introduce a set of new coherence monotones, which we call ``generalized coherence concurrence'' (or coherence $k$-concurrence following), which is organized to witness the coherence number and also has the formal and operational similarity with the generalized  entanglement concurrence. The coherence 2-concurrence from our construction is different from the coherence concurrence recently provided by \cite{Qi}, and we compare the two monotones and $l_1$-norm coherence monotone $C_{l_1}$ quantitatively.
It will be seen that the coherence $k$-concurrence of a mixed state $\r$, denoted as $C_c^{(k)}(\r)$, is nonzero if and only if $r_C(\r) \ge k$, and the state $\r$ can be converted to a state with nonzero entanglement $k$-concurrence via incoherent operations by adding an ancilla system $A$ set in a fixed incoherent state if and only if $C_c^{(k)}(\r)$ is nonzero.

The fact that the coherence concurrence family is well-ordered and has a hierarchy means that the monotones will be useful when we need to know for a quantum system how many bases are how intensively coherent to each other. As an example problem, we exploit the generalized coherence concurrence to understand wave-particle duality in the context of multi-slit interference experiments. It turns out that the concurrence family can capture the path distinguishability more accurately than $l_1$-norm monotone used in \cite{bera, paul2017}. 

This paper is organized as follows. 
In Section \ref{entkcon},  we briefly review the concept of the generalize entanglement concurrence and the coherence number. 
In Section \ref{cohkc}, we define the generalized coherence concurrence and show that it is a coherence monotone family. 
Then we derive the convertibility theorem of  the generalized concurrence between coherence and entanglement,
and compare the second member of the concurrence family, $C_c^{(2)}$, with other coherence monotones such as the coherence concurrence recently presented in \cite{Qi} and $l_1$-norm coherence $C_{l_1}$.
In Section \ref{path}, we calculate a new equation for path distinguiashbility using  $C_c^{(2)}$, and show that the generalized coherence concurrence provides further information on how many slits are distinguisable, which is by the hierarchy property of the coherence concurrence family.
In Section \ref{conclusion}, we summarize our results and present some remaining issues.

\section{REVIEW: THE GENERALIZED ENTALGEMENT CONCURRENCE AND COHERENCE NUMBER }\label{entkcon}
\indent

In this section we briefly review the concepts of the generalized entanglement concurrence and the coherence number.

\subsection*{The  generalized entaglement concurrence}

The  generalized entanglement concurrence monotone is a set of entanglement monotones for $(d\times d)$-systems \cite{gour}. This is the generalization of the $(2\times 2)$-system entanglement concurrence \cite{hill, wooters}.
With a $(d\times d)$-dimensional bipartite pure state $|\p\> = \sum_i\sqrt{\la_i}|i\tilde{i}\>_{AB}$ ($\la_i$'s are the Schmidt coefficients of $|\p\>$), the \emph{$k$-concurrence} ($2\le k \le d$) of $|\p\>$ is defined as  
\begin{align}
\label{ckp}
E_c^{(k)}(|\p\>) &\equiv \Big[\frac{S_k(\la)}{S_k(1/d,1/d, \cdots , 1/d)}\Big]^\frac{1}{k}, \\
S_k(\la) &\equiv \sum_{i_1 < i_2 < \cdots< i_k}\la_{i_1}\la_{i_2}\cdots \la_{i_k}, 
\end{align}
where $S_k(1/d, \cdots, 1/d) =\frac{1}{d^k}\binom{d}{k}$, so $E_c^{(k)}(|\p\>)$ is normalized as $0 \le E_c^{(k)}(|\p\>) \le 1$. Only when $|\p\>$ is maximally entangled  $E_c^{(k)}(|\p\>)$ is equal to $1$.
The $k$-concurrence $E_c^{(k)}(\r)$ for a mixed state $\r$ is defined by convex roof extension:
\begin{align}
&E_c^{(k)}(\r) \equiv \min_{\{p_i,|\p_i\>\}} \sum_i p_i E_c^{(k)}(|\p_i\>) \nn \\
&\Big(\r =\sum_ip_i|\p_i\>\<\p_i|, \quad \sum_ip_i=1, \quad p_i\ge 0\Big).
\end{align}
The $k$-concurrence of $\r$ is nonzero only when the Schmidt number of $\r$ is not smaller than $k$. All the $k$-concurrences with $2 \le k \le d$ consist in the \emph{generalized entanglement concurrence}. This entanglement monotone family is worth investigating since it consists of continuous measures of all possible entanglement dimension, some of which might be useful resources for quantum computation \cite{sentis}.

The last member of the concurrence family is named $G$-concurrence $G_d$, i.e.,  $G_d=E_c^{(d)}$. It has some convenient properties such as multiplicativity which are from the geometric mean form of the concurrence. The  $G$-concurrence also provides a lower bound for the whole $k$-concurrence family:
\begin{align}
E_c^{(2)}(\r) \ge E_c^{(3)}(\r) \ge \cdots \ge E_c^{(d-1)}(\r) \ge E_c^{(d)}(\r) = G_d(\r).
\end{align}
This is useful to analyze some entanglement system such as remote entanglement distribution (RED) protocols \cite{gour}.

The pure state $k$-concurrence can be rewritten in terms of $|\p\>$ that is not Schmidt-decomposed, i.e.,
\begin{align}
|\p\> = \sum_{ij}\p_{ij}|ij\>_{AB},
\end{align} 
which is given by \cite{Chin}
\begin{align}
&E_c^{(k)}(|\p\>) \nn \\
& =d\Bigg[\frac{1}{\binom{d}{k}}\sum_{\substack{i_1< \cdots <i_k \\ j_1 < \cdots <j_k} } 
 \Bigg|\sum_{a_1,\cdots, a_k} \eps_{a_1\cdots a_k} \p_{i_1j_{a_1}} \p_{i_2j_{a_2}}\cdots\p_{i_kj_{a_k}} \Bigg|^2 \Bigg]^{\frac{1}{k}}
\end{align}

\subsection*{Coherence number}
\indent

The coherence rank of a pure state \cite{killoran} is defined as 
\begin{align}
r_C(|\p\>) \equiv \min \Bigg\{  r \Bigg| |\p\> =\sum_{j=1}^{r\le d}\p_j|c_j\> \Bigg\}, 
\end{align}
where the set $\{|c_i\> \}$ is the (classical) referential basis set that is relabeled, and $\forall j:\p_j \neq 0$. A pure state is  nonclassical when $r_C >  1$. There exists a unitary operation $\La$ on $|\p\>$  such that the Schmidt rank of $\La|\p\>$ equals the coherence rank of $|\p\>$.
And the $coherence$ $number$ \cite{Chin} is a generalized concept of the coherence rank, in a similar manner to the Schmidt number \cite{terhal}.

\begin{definition}\label{cohnum}
	The coherence number $r_C(\r)$ for a mixed state $\r$ is defined as 
	\begin{align}
	r_C(\r) \equiv \min_{\{(p_a,|\p_a\>)\}}\max_a\Big[ r_C(|\p_a\>)\Big].
	\end{align}
\end{definition}
For pure states the coherence number is equal to the coherence rank.
The logarithm of the coherence number $\log[\r_C(\r)]$ is a discrete coherence monotone satisfying the axioms (C1-3) listed in Section \ref{cohkc}.

It is proved that the coherence number is a simple criterion for a state $\r$ to be a source for nonzero entanglement $k$-concurrences, i.e., $\r$ can be converted to an entangled state with $E_c^{(k)}(\La[\r]) \neq 0$ ($\La$ is an operation between the given system of $\r$ and and an ancilla system) if and only if $r_C(\r) \ge k$ \cite{Chin}.

\section{GENERALIZED COHERENCE CONCURRENCE}\label{cohkc}
\indent

The coherence resource theory resembles the entanglement resource theory in many aspects, which is considered as evidence that the quantum entanglement is a sort of derivative of coherence. Streltsov $et$ $al.$ \cite{strel} showed that a coherent state can be a resource of a bipartite entangled state through incoherent operations by attaching an ancilla system. A similar process is possible for quantum discord \cite{ma}.

 In this section, we introduce a family of new coherence monotones that is designed to correspond to the  generalized entanglement concurrence. Before tackling the main task, we summarize the axioms that the coherence monotones should fulfill \cite{baum}:\\
 \\
 (C1) Nonnegativity: $C(\r) \ge 0$
 
 (a stronger condition: $C(\r) =0$ if and only if $\r$ is incoherent)\\
 (C2) Monotonicity: $C(\r)$ is non-increasing under the incoherent operations, i.e.,
$C(\La[\r]) \le C(\r)$ for any incoherent operation  $\La$, where $\La: \mathcal{B(H)} \mapsto \mathcal{B(H)}$ admits a set of Kraus operators $\{K_n \}$  such that $\sum_n K_n^\dagger K_n=\mathbb{I}$ and $K_n\d K_n^\dagger$ $\in \mathcal{I}$ for any $\d \in \mathcal{I}$ (the set of incoherent density operators).  \\
(C3) Strong monotonicity: $C$ is non-increasing under  selective incoherent operations, i.e., $\sum_np_nC(\r_n) \le C(\r)$ with $p_n=tr[K_n\r K_n^\dagger]$, $\r_n=K_n\r K_n^\dagger/p_n$ for incoherent Kraus operators $K_n$.\\
(C4) Convexity: $\sum_ip_iC(\r_i) \ge C\Big(\sum_ip_i\r_i\Big)$.
\\

The conditions (C1) and (C2) are the minimal requirements for a quantity  to be a coherent monotone, and a quantity that fulfills (C3) and (C4) naturally fulfills (C2). 

\subsection*{The generalized coherence concurrence as a coherence monotone family}

As mentioned in Section \ref{entkcon}, the entanglement $k$-concurrence of a mixed state $\r$ is nonzero if and only if the Schmidt number of $\r$ is not smaller than $k$. Since the coherence number (rank) is the corresponding quantity in coherence theory to Schmidt number (rank) in entanglement theory, we expect that if  there exists a coherence monotone that corresponds to the entanglement $k$-concurrence, the monotone would have a similar relation with coherence number to the relation of entanglement $k$-concurrence with Schmidt number. We name it \emph{coherence $k$-concurrence}, which consists in the \emph{generalized coherence concurrence} family). The definition of the coherence $k$-concurrence is as follows:
\begin{definition}
\label{cohkcon}	
	 The coherence $k$-concurrence $C_c^{(k)}$ for a pure state $|\p\>=\sum_i\p_i|i\>$ is defined as

\begin{align}
C_c^{(k)}(|\p\>) &=d\Bigg( \frac{1}{\binom{d}{k}} \sum_{i_1<i_2<\cdots <i_k}\Big| \p_{i_1}^2\p_{i_2}^2\cdots\p_{i_k}^2\Big|\Bigg)^\frac{1}{k}.,
\end{align}
and  $C_c^{(k)}(\r)$ for a mixed state $\r$ is defined by convex roof extension:
\begin{align}
	&C_c^{(k)}(\r) \equiv \min_{\{p_a,|\p_a\>\}} \sum_a p_aC_k(|\p_a\>) \nn \\
	&\Big(\r =\sum_a p_a|\p_a\>\<\p_a|, \quad \sum_a p_a=1, \quad p_a\ge 0\Big).
\end{align}
A family of generalized coherence concurrence consists of $C_c^{(k)}$ with $2\le k \le d$.
\end{definition}

The normalization factor is multiplied so that the generalized coherence  concurrence has a similar inequality order to the generalized entanglement concurrence, i.e.,
\begin{align}\label{order}
C_c^{(2)}(\r) \ge C_c^{(3)}(\r) \ge \cdots \ge C_c^{(d-1)}(\r) \ge C_c^{(d)}(\r),
\end{align}
which is straightforward by Maclaurin's inequality.

\begin{theorem}
$C_c^{(k)}$ is a coherence monotone that satisfies (C1) to (C4).
\end{theorem}
\begin{proof}
 The fulfillment of (C1) and (C4) is clear by definition, the strong condition of (C1) does not hold though (see Theorem \ref{cohconnum} below). (C2) is satisfied if (C3) and (C4) are satisfied, so what remains to be proven is (C3).
 It is proved in  Appendix \ref{C3pure} that (C3) is fulfilled for pure states, i.e., 
\begin{align}
\label{C3p}
	\sum_n p_nC_c^{(k)}(|\p_n\>) \le C_c^{(k)}(|\p\>).
\end{align}
  Then $C_c^{(k)}$ also satisfies (C3) for the convex roof extension of the quantity to mixed states (see Appendix A.1 of \cite{Qi}).

\end{proof}  

\subsection*{The conversion of concurrence from coherence into entanglement}\label{cohent}
\indent

In this subsection, we examine the entanglement convertibility theorem of the generalized concurrence monotone and show that the generalized coherence concurrence is an entanglement-based monotone \cite{strel}. The referential entanglement monotone is the generalized entanglement concurrence as expected. 

To obtain the convertibility theorem, we first clarify the relation between the coherence $k$-concurrence and the coherence number, which is basically identical to the relation between the entanglement $k$-concurrence and the Schmidt number.

\begin{theorem} \label{cohconnum}
	For a state $\r$,
	the coherence $k$-concurrence $C_c^{(k)}(\r)$  is nonzero if and only if the coherence number $r_C(\r)$ is not smaller than $k$.	
\end{theorem}
\begin{proof}
	$\Longrightarrow$: Supposing $\{p_a,|\p_a\> \}$ is the optimal decomposition of $\r$ for $C_c^{(k)}$, the condition  $C_c^{(k)}\neq 0$ means that there exists at least one decomposing pure state $|\p_a\>$ that satisfies $C_c^{(k)}(|\p_a\>)\neq 0$. So  $r_C(|\p_a\>) \ge k$ by Definition \ref{cohkcon}, which goes to $r_C(\r) \ge k$ by Definition \ref{cohnum}.
	\\
	$\Longleftarrow$: Suppose  $C_c^{(k)}(\r)=0$. Then there exists a decomposition $\{p_a,|\p_a\>\}$ that satisfies $\forall a:r_C(|\p_a\>)< k$. So we have $r_C(\r)<k$.  
\end{proof}

Now with Theorem \ref{cohconnum} we can obtain the conversion theorem for each coherence $k$-concurrence quite simply: 
\begin{theorem}
	A state $\r^s$ can be converted to a state of nonzero entanglement $k$-concurrence via an incoherent operation by appending an ancillar system $A$ which is set in a referential incoherent  state $|1\>\<1|^A$ if and only if the coherence $k$-concurrence is nonzero.
\end{theorem}
\begin{proof} This statement is derived using the coherence number, which links the generalized concurrences for coherence and entanglement theory. Theorem 3 of \cite{Chin} and Theorem \ref{cohkcon} gives
	\begin{align}
	\exists \La^{SA}: E_c^{(k)}(\La^{SA}[\r\otimes |1\>\<1|^A])  \neq 0  & \iff r_C(\r) \ge k  \nn \\
	& \iff C_c^{(k)}(\r^s)\neq 0.
	\end{align}
\end{proof}

Now we consider what happens when the incoherent operation on the bipartite system is a unitary operation that transforms the coherence rank of a pure state to the Schmidt rank of the bipartite pure state,
\begin{align}
\label{u}
U\equiv \sum_{i=1}^{d}\sum_{j=i}^{d}|i\>\<i|^S\otimes |i\oplus (j-1)\>\<j|^A,
\end{align}
where $\oplus$ represents an addition modulo $d$. This is often called the generalized CNOT operation. And 
$|\p\>^S=\sum_{i=1}^{d}\p_i|i\>$ goes to $|\p\>^{SA}=\sum_{i=1}^{d}\p_i|ii\>$ under $\La^{SA}_u$ ($\La^{SA}_u: \r \mapsto U[\r^S\otimes |1\>\<1|^A] U^\dagger $).

Then we have
\begin{align}
E_c^{(k)}(|\p\>^{SA})& = d\Bigg[ \frac{1}{\binom{d}{k}}\sum_{i_1 < \cdots i_k} | \p_{i_1} \cdots \p_{i_k}|^2 \Bigg]^{\frac{1}{k}} \nn \\
& = C_c^{(k)}(|\p\>^S).	
\end{align}
for a pure state $|\p\>$, and this equality holds for a mixed state $\r$ by the definition of convex roof extension:
\begin{align}
E_c^{(k)}(\La^{SA}_u[\r^s\otimes |1\>\<1|^A]) = C_c^{(k)}(\r^s).	
\end{align}
So we can say that the generalized coherence concurrence is a kind of entanglement-based coherence monotone (Eq. (12) of \cite{strel}).

\subsection*{Comparison of $C_c^{(2)}$ with the coherence concurrence $C_c$ and $l_1$-norm coherence $C_{l_1}$}\label{c2cccl1} 

In the entaglement theory, the 2-concurrence in the generalized concurrence family is equal to the entanglement concurrence presented by \cite{hill, wooters}. In the coherence theory, we can compare 2-concurrence with the coherence concurrence $C_c$ recently presented by \cite{Qi}, which is defined as
\begin{align}
C_c(|\p\>) =2\sum_{j<k} |\p_j\p_k| 
\end{align}
for a pure state $|\p\> = \sum_i\p_i|i\>$ and  $C_c(\r)$ for a mixed state is the convex roof extension of the pure state monotone.

We have
\begin{align}
 C_c^{(2)}(|\p\>)& =\frac{d}{\binom{d}{2}^\frac{1}{2}}(\sum_{j<k} |\p_j|^2|\p_k|^2)^\frac{1}{2},
\end{align}
and the following relation between $C_c^{(2)}$ and $C_c$ holds:
\begin{theorem}\label{c2rel}
\begin{align}
 \frac{1}{d-1} C_c(\r) \le C_c^{(2)}(\r) \le \sqrt{\frac{d}{2(d-1)}}C_c(\r).
\end{align}	
\end{theorem}
\begin{proof}
 For pure state case, the right inequality is trivial and the left inequality comes from the Newton's inequality of elementary symmetric polynomials. By convex roof extension, we obtain Theorem \ref{c2rel}.
\end{proof}
This relation is interesting for some reasons. First, it supports the claim in \cite{rana, rana2} that  $C_{l_1}$ is analogous to entanglement negativity. Indeed, considering $C_c$ is equal to $C_{l_1}$ for pure states, we have
\begin{align}
 \frac{1}{d-1} C_{l_1}(|\p\>) \le C_c^{(2)}(|\p\>) \le \sqrt{\frac{d}{2(d-1)}}C_{l_1}(|\p\>).
\end{align}
There exists the same form of inequality between the 2-concurrence and the negativity in entanglement theory \cite{eltschka}. Second, combining Theorem \ref{c2rel} with Eq. \eqref{order}, $C_c(\r)$ imposes an uppper bound for the whole coherence $k$-concurrence family, i.e.,
\begin{align}\label{order}
\sqrt{\frac{d}{2(d-1)}}C_c(\r) &\ge  C_c^{(2)}(\r)\nn \\
         & \ge C_c^{(3)}(\r) \ge \cdots \ge C_c^{(d-1)}(\r) \ge C_c^{(d)}(\r).
\end{align}

As an additional discussion, $C_c$ is in general not smaller than $l_1$-norm coherence monotone, i.e.,
\begin{align}
C_c(\r) \ge  C_{l_1}(\r) \equiv 2\sum_{j<k}|\r_{jk}|
\end{align}
\cite{Qi}, but there exists a necessary and sufficient condition for $C_c$ of a mixed state to be equal to $C_{l_1}$.

\begin{theorem}\label{CcCl}
For $d \ge 3$, $C_c(\r)$ and $C_{l_1}(\r)$ coincide   if and only if the state satisfies 
\begin{align}
\label{restriction}
&\qquad \frac{\r_{ij}\r_{jk}\r_{ki} }{|\r_{ij}\r_{jk}\r_{ki} |} = 1 \nn \\
& (\textrm{no summation over } i,j, k \textrm{ and } i\neq j\neq k )
\end{align}
for all non-zero components of $\r$.  For $d=2$, they always coincide. 
\end{theorem}
The proof is given in Appendix \ref{CcCl1}.
We can see that the equality 
$C_c(\r)=C_{l_1}(\r)$ always holds  for a real symmetric state $\r$.

\section{$C_c^{(k)}$ AND PATH DISTINGUISHABILITY}\label{path}
We expect that the generalized coherence concurrence is useful for some quantum systems about which we want to know both how many bases are coherent with each other and how much they are coherent to each other. As an example, here we try to delve into the  \emph{path distinguishability problem} in multi-slit experiments using $C_c^{(k)}$.    

Wave-particle duality is an ironic but intriguing property
of quantum theory. There have been efforts to understand the complementary principle in the context of two-path interference experiments \cite{wootters1979, green19, engl1996}, in which the duality is quantitatively expressed, e.g., as the Englert-Greenberger-Yasin (EGY) relation.
 The investigation has gone further to multi-path cases \cite{jaeger1995, durr2001, bimonte2003, englert2}. 
 
   Based on the
natural idea that coherence is a representative wave-like property, a new duality relation for
general $d$-slit interference was obtained using $C_{l_1}$ \cite{bera, paul2017}:
\begin{align}
\mathcal{D}_Q +\frac{1}{d-1}C_{l_1} \le 1,
\end{align}
where $\mathcal{D}_Q$ is a path distinguishability, (representative of particle-like  property) based on UQSD (the unambiguous quantum state discrimination). The inequality is saturated  when the quanton (wave-particle-like quantum system) is pure. 

 But as $C_{l_1}$ is non-zero if and only if the coherence number is not smaller than 2, what we can say with $C_{l_1}$ is just whether the quanton has wave-like property or not. We claim that the generalized coherence concurrence provides further information on how many slits are unambiguously identified.
 
  We first show that coherence 2-concurrence $C_c^{(2)}$ gives a sharper bound for $\mathcal{D}_Q$ and then discuss the relation between the amount of $C_c^{(k)}$ for each $k$ and the number of completely distinguishable slits.

\subsection*{Path distinguishability revisited with $C_c^{(2)}$}
We tackle the problem by first considering $d$-slit interference of pure quantons, the state of which is expressed with $d$ basis states $\{|\p_i\>\}_{i=1}^{d}$ as
 
\begin{align}
|\P\> = \sum_ic_i|\p_i\>,
\end{align}
where the bases are orthonormal since they represent  well seperated different slits and $\sum_i |c_i|^2=1$.
To measure the quanton, we need to let a detector interact and correlate as follows:
\begin{align}
|\P\> \mapsto  \sum_ic_i|\p_i\>\otimes |0\>_D \mapsto  \sum_ic_i|\p_i\>\otimes |i\>_D.
\end{align}
The states are normlized, $\<\p_i|\p_i\>=\<i|i\>_D=1$. 
But as $|i\>_D$'s are not necessarily orthogonal, we express them with an orthonormal basis set $\{|a\>\}$ and $\{|\phi_i\>\}$ as
\begin{align}
|i\>_D = \phi_i|\phi_i\> +\sum_a\sqrt{p_a}q_a^i|a\> \equiv \phi_i|\phi_i\> +|q^i\>,
\end{align}
where $\sum_a p_a=1$,  $\<\phi_i|\phi_j\> \neq 0$ only when $i =j$, and $ \<\phi_i|a\>=0$ for all $j$ and $a$.
The normalization condition gives $\forall i: $ $|\phi_i|^2 + \sum_a p_a |q_a^i|^2 =1$.
And the reduced density matrix of the quanton is given by
\begin{align}
\label{r^s}
\r^s &= tr_D\Big(|\P\>\<\P|\Big)\nn \\
   &=\sum_{i,j} (\sum_a p_a c_iq^i_a c^*_jq^{j*}_a +\delta_{ij}c_ic_j^*|\phi_i|^2 ) |\p_i\>\<\p_j| \nn \\
   &= \sum_i |c_i\phi_i|^2 |\p_i\>\<\p_i| + \sum_a \Bigg(p_a \sum_i|c_iq_a^i|^2 \Bigg) |\p^a\>\<\p^a|,
\end{align}
where $|\p_a\> \equiv \sum_ic_iq^i_a |\p_i\>/ \sqrt{\sum_i|c_iq_a^i|^2}$  is normalized. This is one way of pure state decompostion for $\r^s$, so we have 
\begin{align}
\label{pureineq}
C_c^{(k)}(\r^s) \le   \sum_a \Bigg(p_a\sum_i|c_iq_a^i|^2\Bigg) C_c^{(k)}(|\p_a\>).
\end{align}
by the definition of $C_c^{(k)}$. We can adjust $\{p_a, q_a^i \}$ so that the inequality is saturated.
 A special case is when 
\begin{align}
|i\>_D = |a\>,
\end{align}
,i.e., $|i\>_D$ are the same for all $i$. Then Eq. \eqref{r^s} becomes pure,
\begin{align}
\r^s =  \sum_{i,j} c_ic^*_j|\p_i\>\<\p_j|,
\end{align}
which means that the measurement plays no role for the quanton system.

 To obtain the path distinguishability, we divide measurement operations into two groups as
\begin{align}
\label{projection}
\hat{A}_m|i\>_D \propto & \quad|\phi_i\> ,\nn \\
\hat{B}_m|i\>_D \propto &  \quad |q^i\> 
\end{align}
so that $\hat{A}_m$ and $\hat{B}_m$ represent successful and failure transformations for distinguishability respectively with the restriction
 $\sum_m\Big( \hat{A}_m\hat{A}_m^\dagger +\hat{B}_m\hat{B}_m^\dagger\Big) = \mathbb{I}$.
The success and failure probability are defined as
\begin{align}
& P_d=\sum_{i}|c_i|^2 \sum_{m}\<i| \hat{A}_m^{\dagger}\hat{A}_m|i\>_D, \nn \\
& Q_d=\sum_{i}|c_i|^2 \sum_{m}\<i| \hat{B}_m^{\dagger}\hat{B}_m|i\>_D. \qquad (P+Q=1) 
\end{align}
Then using Cauchy-Schwarz inequality (see Eq. (9) of \cite{qiu2002}), we have 
\begin{align}
\label{q^2}
Q_d^2 \ge \frac{d}{d-1}\sum_{i\neq j} |c_i|^2|c_j|^2 & \<i| \sum_m \hat{B}_m^{\dagger}\hat{B}_m|i\>_D \nn \\
 & \times \<j| \sum_m \hat{B}_m^{\dagger}\hat{B}_m|j\>_D
\end{align}
Eq. \eqref{projection} gives 
\begin{align}
 \<i| \sum_m \hat{B}_m^{\dagger}\hat{B}_m|i\>_D & =\sum_{a,b}\sqrt{p_a p_b}q_a^{i*}q_b^i \<a| \sum_m \hat{B}_m^{\dagger}\hat{B}_m|b\> \nn \\
  &=\sum_{a,b}\sqrt{p_a p_b}q_a^{i*}q_b^i \<a| \Big( \mathbb{I} - \sum_m \hat{A}_m^{\dagger}\hat{A}_m\Big)|b\> \nn \\
   &= \sum_a p_a |q_a^i|^2,
\end{align}
and Eq. \eqref{q^2} is rewritten as
\begin{align}
\label{qfin}
Q_d &\ge  \Bigg[\frac{d}{d-1}\sum_{i\neq j} |c_i|^2|c_j|^2 \sum_{a,b}\Big(p_a|q_a^i|^2\Big) \Big(p_b|q_b^j|^2\Big) \Bigg]^\frac{1}{2} \nn \\
 & \ge \sum_a p_a \Bigg[\frac{d}{d-1}\sum_{i\neq j} |c_i|^2|c_j|^2 |q_a^i|^2|q_a^j|^2 \Bigg]^\frac{1}{2} \nn \\
 & \ge C_c^{(2)}(\r^s).
\end{align}
The second inequality comes from Eq. \eqref{subadd} and the third from Eq.  \eqref{pureineq}.
So the success probability is bounded by
\begin{align}
P_d \le 1- C_c^{(2)},
\end{align}
and the path distinguishability $\mathcal{D}_Q$, the upper bound of $P_d$, is given by
\begin{align}
\mathcal{D}_Q = 1-C_c^{(2)} \le 1-\frac{1}{d-1} C_{l_1} 
\end{align}
This is a tighter upper bound than that presented in \cite{bera}.

Now we move on to mixed quanton case, in which the quanton system has some degree of interation with the environment. The mixed state density matrix is expressed as 
\begin{align}
\r^{sd} = \sum_x \la_x \sum_{i,j} \chi_x^i \chi_x^{j*} |\p_i\>\<\p_j|\otimes |i\>\<j|_D.
\end{align}
After partial-tracing the detector, the reduced density matrix is given by 
\begin{align}
\r^s = & \sum_{a,x}p_a\la_x \sum_{i,j}\chi^i_x q_a^i \chi^{j*}_{x} q_a^{j*}|\p_i\>\<\p_j| \nn \\
& + \sum_i \Bigg( \sum_x |\chi_x^i|^2\Bigg) |\phi_i|^2 |\p_i\>\<\p_j|,
\end{align} 
This is a pure state decompostion of $\r^s$ and a similar ineqaulity to Eq.  \eqref{pureineq} holds.
Since $|i\>_D$ appears with probability $\sum_x\la_x|\chi^i_x|^2$, the failure probability for the mixed state is bounded below as
\begin{align}
Q_d  &\ge  \sum_a p_a \Bigg[ \frac{d}{d-1}\sum_{i\neq j} \sum_x\la_x|\chi^i_x|^2 \sum_y\la_y|\chi^i_y|^2  |q_a^i|^2 |q_a^j|^2 \Bigg]^\frac{1}{2} \nn \\
 & \ge \sum_{a,x} p_a\la_x \Bigg[ \frac{d}{d-1}\sum_{i\neq j} |\chi^i_x q_a^i|^2|\chi^j_x q_a^j|^2 \Bigg]^\frac{1}{2} \nn \\
 & \ge C_c^{2}(\r^s).
\end{align}
So $\mathcal{D}_Q$ for mixed states is also given by
\begin{align}
\mathcal{D}_Q = 1-C_c^{(2)}, 
\end{align}
which is more accurate result than the inquality $\mathcal{D}_Q \le 1-\frac{1}{d-1}C_{l_1}$ for mixed states given in \cite{bera}.

Summarizing, we obtained the same form of the path distinguishability for both pure and mixed quanton systems with $C_c^{(2)}$. On the other hand, $l_1$-norm presents less tight bound for pure systems and inequality for mixed states. 

\subsection*{$C_c^{(k)}$ and the number of distinguishable slits} 

Now we think of a quanton state $\r^{s}$ with $C_c^{(k+1)} = C_c^{(k+2)} =  \cdots = C_c^{(d)}=0$, or $r_C(\r^s)=k$ equivalently.
In this case we can express the detector states, without loss of generality, as
\begin{align}
& |1\>_D = \phi_1|\phi_1\> +|q^1\>, \quad |2\>_D = \phi_2|\phi_2\> +|q^2\>, \quad \cdots , \nn \\
& |k+1\>_D = |\phi_{(k+1)}\>, \quad \cdots , \quad |d\>_D = |\phi_d\>,
\end{align}
which we can see from Eq. \eqref{r^s}.
With measurement operators $A_m $ and $B_m$ in \eqref{projection}, we can not receive confusing information from $(k+1)$ to $d$-th slit. So we can state that \emph{if $C_c^{(k+1)} =0$ for the quanton state $\r^s$ then there exist $(d-k)$-slits that we can identify unambiguously.}

\section{Conclusions}\label{conclusion}
\indent
In this study, we introduced a family of new coherence monotones, generalized coherence concurrence, which has a close interrelationship with the coherence number and  the generalized entanglement concurrence.
  We then compared the coherence 2-concurrence with the coherence concurrence of \cite{Qi} and $l_1$-norm coherence, which supports the assumption that the operational role of $l_1$-norm in coherence quantitative theory is that of negativity in entanglement theory. An example for the applications of $C_c^{(k)}$ to quantum systems was path distinguishability problems. We obtained a sharper equation for the path distinguishability, and gave the relation between the coherence number of the system and the number of identifiable slits.  

One of the remaining problems is to measure $C_c^{(k)}$ as the path distinguishability in the experimental interference pattern, as $C_{l_1}$ in \cite{paul2017}. We also guess the other members of the coherence family have accurate relations with some quantities in multi-slit interference. Comparison of the coherence monotones obtained from the differential Chernoff bound \cite{calsam, bisw} might present a clue to this problem.  In a broader sense, the relation of four quantities--- the coherence number, the Schmidt number, the generalized entanglement and coherence concurrence--- would have an interesting structure to pursue further in both mathematical and practical directions.  We expect the generalized coherence concurrence will be useful for any quantum phenomena in which not only the amount but also the order of coherence is crucial. 

\begin{acknowledgments}
The author is grateful to Prof. Jung-Hoon Chun for his advice during the research. This was supported by Basic Science Research Program through the National Research Foundation of Korea funded by the Ministry of Education(NRF-2016R1D1A1B04933413).

\end{acknowledgments}

\appendix

\section{Proof of \eqref{C3p}}\label{C3pure} 

 For a pure state $|\p\>$, we have
\begin{align}
&\sum_n p_nC_c^{(k)}(|\p_n\>)\qquad \Big(|\p_n\> \equiv \frac{K_n|\p\>}{\< \p|K_n^\dagger K_n|\p \>^\frac{1}{2}} \Big) \nn \\
& = d \sum_n \Big(\frac{(d-k)!}{d!} \nn \\
&\qquad\quad \times \sum_{i_1 \neq \cdots \neq i_k}\Big| \Big(\sum_{j_1} K_n^{i_1j_1}\p_{j_1} \Big)^2\cdots \Big(\sum_{j_k}K_n^{i_kj_k}\p_{j_k}\Big)^2 \Big| \Big)^{\frac{1}{k}} 
\end{align}
Then 
\begin{align}
&\frac{1}{d}\Big(\frac{d!}{(d-k!)}\Big)^{\frac{1}{k}}\sum_n p_nC_c^{(k)}(|\p_n\>)  \nn \\
& = \sum_n \Big(\sum_{i_1 \neq \cdots \neq i_k}\Big| \sum_{j_1\neq \cdots \neq j_k} (K_n^{i_1j_1}\p_{j_1})^2 \cdots (K_n^{i_kj_k}\p_{j_k})^2 \Big| \Big)^{\frac{1}{k}}
\end{align}
from the fact that $K_n$ is expressed as $K_n=\sum_i c_n^i |s_i\>\<i|$ where $|s_i\> $ is a re-ordered vector of the referential
basis index \cite{winter}.  Using  the relation
\begin{align}
\label{subadd}
S_k(\la_1)^{\frac{1}{k}} +S_k(\la_2)^\frac{1}{k} \le S_k(\la_1+\la_2)^{\frac{1}{k}}
\end{align}
of elementary symmetric  polynomials $S_k$ \cite{gour}, we have the following inequality:
\begin{align}
& \sum_n \Big(\sum_{i_1 \neq \cdots \neq i_k}\Big| \sum_{j_1\neq \cdots j_k} (K_n^{i_1j_1}\p_{j_1})^2 \cdots (K_n^{i_kj_k}\p_{j_k})^2 \Big| \Big)^{\frac{1}{k}} \nn \\
& \le \sum_n \Big( \sum_{j_1 \neq \cdots  \neq  j_k }\Big|    \p_{j_1}^2 \cdots \p_{j_k}^2\Big| 
\sum_{i_1 \neq \cdots  \neq  i_k}\Big| K_n^{i_1j_1}\cdots K_n^{i_kj_k}\Big|^2 \Big)^{\frac{1}{k}} \nn \\
& \le \Big( \sum_{j_1 \neq \cdots  \neq  j_k}\Big| \p_{j_1}^2 \cdots \p_{j_k}^2\Big| \sum_{n_1, i_1}\Big| K_{n_1}^{i_1j_1}\Big|^2 \cdots \sum_{n_k,i_k} \Big| K_n^{i_kj_k} \Big|^2 \Big)^{\frac{1}{k}}.
\end{align}
Since 
\begin{align}
\sum_{n,i}|K_n^{ij}||K_n^{ij}| =\sum_n|(c_n^j)^2 \<s_j|(\sum_i|i\>\<i|)|s_j\>| =1,
\end{align}
we finally have
\begin{align}
\sum_n p_nC_c^{(k)}(|\p_n\>) \le C_c^{(k)}(|\p\>).
\end{align}

\section{Proof of Theorem \ref{CcCl}}\label{CcCl1}

\noindent 1. $d\ge 3 $ case:
\\
a. $\frac{\r_{ij}\r_{jk}\r_{ki}}{|\r_{ij}\r_{jk}\r_{ki} |} = 1 \Longrightarrow C_c(\r) = C_{l_1}(\r)$
	
	The hermiticity of a quantum state $\r$ is explicitly expressed as
	\begin{align}
	\label{as}
	\r_{ij}=|\r_{ij}|e^{i\th_{ij}}, \quad \textrm{with} \quad |\r_{ij}|=|\r_{ji}|, \quad \th_{ij} =-\th_{ji}.
	\end{align}
	The condition \eqref{restriction} restricts the phases of $\r_{ij}$'s  as 
	\begin{align}
	\label{rot}
	\th_{ij} +\th_{jk}+\th_{ki}=0.
	\end{align}
	Then the general solution of \eqref{as} and \eqref{rot} is given with new real variables $\th_i$ by
	\begin{align}
	\th_{ij} = \th_i -\th_j
	\end{align}
	(See, e.g., p40 of \cite{siegel}). So it is always possible to decompose $\r$ with
	\begin{align}
	|\tilde{\p}_a\>= \sum_{i}\tilde{\p}_a^i|i\> =|\tilde{\p}_a^i|e^{i\th_i}|i\> 
	\end{align}
	as
	\be
	\r_{jk} =\sum_{a}\tilde{p}_a\tilde{\p}^j_a(\tilde{\p}^k_a)^* = \Bigg(\sum_{a}\tilde{p}_a|\tilde{\p}^j_a\tilde{\p}^k_a|\Bigg) e^{i(\th_j-\th_k)}.
	\ee
	Then $C_{l_1}$ under this decomposition is
	\be
	C_{l_1}(\r)=2\sum_{j<k}\sum_{a}\tilde{p}_a|\tilde{\p}^j_a\tilde{\p}^k_a|,
	\ee
	and we have
	\be
	C_{c}(\r) \le \sum_{a}\tilde{p}_a C_c(|\tilde{\p}_a\>) =C_{l_1}(\r).
	\ee
	Since the inequality $ C_c(\r) \ge C_{l_1}(\r)$ always holds, we have  $ C_c(\r) = C_{l_1}(\r)$.
	\\ 
	\\
b. $ C_c(\r) = C_{l_1}(\r) \Longrightarrow \frac{\r_{ij}\r_{jk}\r_{ki}}{|\r_{ij}\r_{jk}\r_{ki} |} = 1$
	
	Suppose that the decomposition of $\r $ as $\r =\sum_{\a}p_a|\p_\a\>\<\p_\a|$ gives the minimal value for $C_c(\r)$. Then we have
	\be
	C_c(\r) = 2\sum_{j<k} \sum_\a p_\a |\p^\a_j\p^\a_k|.
	\ee
	On the other hand, $C_{l_1}(\r)$ expressed with this decomposition is
	\be
	C_{l_1}(\r)=2\sum_{j<k} \Big|\sum_a p_\a\p^\a_j(\p^{\a}_k)^*\Big|.
	\ee 
	Since the inequality
	\be
	\label{ine}
	\sum_a p_a|\p_j^a\p_k^a| \ge  \Big|\sum_a p_a\p^a_j(\p^{a}_k)^*\Big| 
	\ee
	holds for all $(j,k)$, the condition for $C_{l_1}(\r)$ and $C_c(\r)$ to be the same value is that \eqref{ine} is saturated for all  $(j,k)$. This is equivalent to 
	\begin{align}
	\th_j^a -\th_k^a =\th_j-\th_k = \th_{jk},\quad \forall a.
	\end{align} This completes the proof for $d\ge 3$.
	\\
	\\
2.	$d=2$ case:
	
	Now $\r_{12}$ can always be written as $|\r_{12}|e^{i(\th_1-\th_2)}$, so following the argument in (a) of $d\ge 3$ case, the coherence concurrence is equal to $l_1$-norm coherence in $d=2$. 

\bibliography{gconpath}

\end{document}